\documentclass[12pt]{iopart}

\usepackage[all]{xy}
\usepackage{iopams}
\usepackage{amsthm}
\newcommand{\comp}{\mbox{\small $\circ$}}
\newcommand{\tcdot}{\hspace{-0.05cm}\cdot \hspace{-0.05cm}}
\newcommand{\mt}[1]{\mbox{\tiny$#1$}}

\theoremstyle{plain}
\newtheorem{thm}{Theorem}[section]

\newtheorem{lem}[thm]{Lemma}
\newtheorem{proposition}[thm]{Proposition}
\theoremstyle{definition}
\newtheorem{definition}[thm]{Definition}
\newtheorem{example}[thm]{Example}
\theoremstyle{remark}
\newtheorem{remark}[thm]{Remark}

\begin{document}

\title[On the geometry of quantum indistinguishability]
{On the geometry of quantum indistinguishability}

\author{A F Reyes-Lega}

\address{Departamento de F\'isica, Universidad de los Andes,
AA 4976, Bogot\'a D.C., Colombia} \ead{anreyes@uniandes.edu.co}
\begin{abstract}
An algebraic approach to the study of quantum mechanics on configuration spaces with a finite fundamental
group is presented. It uses, in an essential way, the Gelfand-Naimark and Serre-Swan equivalences and thus
allows one to represent  geometric properties of such systems in algebraic terms. As an application, the problem
of quantum indistinguishability is reformulated in the light of the proposed approach. Previous attempts
aiming at a proof of the spin-statistics theorem in non-relativistic quantum mechanics are explicitly recast
in the global language inherent to the presented techniques. This leads to a critical discussion of
single-valuedness of wave functions for systems of indistinguishable particles.   Potential applications of the methods presented in
this paper to problems related to  quantization, geometric phases and  phase transitions in spin systems are
proposed.
\end{abstract}

\submitto{\JPA}
\maketitle

\section{Introduction}
As is well-known, there are several proofs of the \emph{spin-statistics theorem} that are valid for quantum
fields. The first versions of the theorem, due to Fierz \cite{Fierz1939} and Pauli \cite{Pauli1940}, date
back to the forties and were valid only for free relativistic quantum fields. The developments that
eventually led to the establishment of the theorem in the general case have been described in the book by
Duck and Sudarshan \cite{Duck1997}. The modern proofs of the theorem, in the framework of
(axiomatic/algebraic) quantum field theory, are described in the books by Streater and Wightman
\cite{WightmanCPT:00} and by Haag \cite{Haag1996}. Although these proofs use relativistic invariance in an
essential way, there are other alternative approaches that do not make use of the full Poincar\'e covariance
of the theory. One of them, based on Schwinger's action principle, has been pioneered by Sudarshan
\cite{Duck1997,Sudarshan1975}.

 A more radical point of view, that has been explored for several years, is based on
the idea that it might be possible to explain the spin-statistics connection from within (nonrelativistic)
quantum mechanics. There have been  many proposals along this line of thought.  One of them, proposed by
Balachandran  and coworkers \cite{Balachandran:90}, uses classical configuration spaces of identical
particles and resembles the work of Finkelstein and Rubinstein on kink configurations for nonlinear fields
\cite{Finkelstein1968}. Even though their approach does not make direct use of quantum field theory, the idea
of pair creation/annihilation is incorporated indirectly in the topology of the configuration space. A
similar approach has also been developed, independently, by Tscheuchner \cite{Tscheuschner1989} (see also
\cite{Mickelsson1984}).

More recently,  Berry and Robbins \cite{Berry1997} have developed a new approach in which no use of
relativity or quantum field theory is made. The basic idea of Berry and Robbins is to construct an operator
that implements the simultaneous and continuous exchange of the particles' spin and position labels. An
essential part of their approach is the imposition of single-valuedness of the total wave function under
exchange, as a way to incorporate indistinguishability without using the symmetrization postulate. Their construction has the virtue of being very explicit and of giving the correct statistics sign,
for any value of the spin. Unfortunately (as realized by the same authors) this construction is not the only
one meeting their requirements: There are other possibilities, allowing for a violation of the
spin-statistics connection \cite{Berry2000}.
     In spite of the fact that it does not provide a ``proof'' of the spin-statistics theorem,
     the Berry-Robbins approach \cite{Berry1997,Berry2000,Berry2000a,Berry2008} has received a lot of attention
     \cite{Twamley1997, Sudarshan2003, Chru'sci'nski2004, Wightman2000} and has inspired new questions and developments,
     both in mathematics \cite{Atiyah2001,Atiyah2002,Atiyah2002a} and physics
     \cite{Anandan1998,Peshkin2003,Harrison2004,Papadopoulos2004,Peshkin2006,Papadopoulos2010}.

     Among these, a proposal to study the spin-statistics connection using
     tools from noncommutative geometry was made by Paschke \cite{Paschke2001}. Using only the $SU(2)$ symmetry
     of the algebra $C(S^2)$ of continuous functions on the sphere, he proved that the subspace consisting of
     all \emph{odd} functions in $C(S^2)$ can be regarded, via the Serre-Swan equivalence, as the space of
     sections of a line bundle over the projective space $\mathbb R P^2$. Using this result, we recognized
     that the exchange-rotation operator of Berry and Robbins could be rewritten as a projector. This projector turned
     out to be given precisely by direct sums of  the  projectors constructed by Paschke \cite{Papadopoulos2004}.
      These results, in turn, led to
     a critical discussion of single-valuedness and its role for the spin-statistics connection
     \cite{Papadopoulos2004,Papadopoulos2010}. The discussion, though, was based on results that
     are valid only for 2 particles.
Therefore, the main purpose of the present paper is to develop and generalize the framework on which the
discussions in \cite{Papadopoulos2004} and \cite{Papadopoulos2010}  are based.
 Although the main motivation for the present paper
has been the spin-statistics connection, the results obtained fit in the more general context of quantum
mechanics on general configuration spaces. As explained in the last section, there are potential applications
of the techniques developed in this paper to the study of certain quantization problems, to problems related
to Berry phases and to the study of spin systems.

Before embarking on the technical part of the paper, I want to make some comments regarding the status of
what can be called the \emph{``configuration space approach''} to spin-statistics, in view of the opinion
among several authors that any attempt at explaining the connection without relativistic quantum field theory
is doomed to fail (cf. \cite{Wightman2000,Allen2003}).

To a big extent, the interest in alternative proofs of the
spin-statistics theorem in the last years increased due to
Neuenschwander \cite{Neuenschwander1994} who,  motivated by a -now
famous- remark of Feynman, asked for a simple explanation of
spin-statistics. Some direct responses soon appeared. Currently,
numerous attempts at simple proofs of the theorem can be found in
the literature. Unfortunately, a significant part of this
production is based on  assumptions and methods lacking either a
clear  physical justification or a sound mathematical basis, this
leading (in some cases) to rather superficial treatments of the
problem. Critiques to such oversimplified approaches can be found,
e.g.,  in Hilborn \cite{Hilborn1995}, Romer \cite{Romer2002} and
Duck and Sudarshan \cite{Duck1997,Duck1998a}.

It seems, therefore, that Feynman's demand for simplicity -in a straightforward sense- cannot be met. But one
should also  bear in mind that many of the nonstandard  approaches have a theme in common, one
that touches upon an important part of quantum theory: The study of quantum mechanics on general
configuration spaces \cite{Isham1984}. This more general problem  touches on the very foundations of quantum mechanics while connecting representation theory with geometry, topology and measure theory, as can be seen, for instance, from the work of Mackey \cite{Mackey1978}. Moreover this general problem has attracted the attention of physicists
ever since the early days of quantum theory. For example, the early investigations by Pauli
\cite{Pauli:39,Pauli1958} were motivated by the search for a single-valuedness criterion for wave functions.
The unavailability of a classical analogue for spin led Bopp and Haag \cite{Bopp:50} to study the quantum
theory of the rotation group, regarded as a configuration space. The crucial contributions made by  Schulman
\cite{Schulman:68} and by Laidlaw and DeWitt \cite{Laidlaw1971} were based on the observation (due to
Schulman) that the path integral in a multiply connected space has to be modified with respect to its usual
form. Nowadays it is clear that the topology of the configuration space plays a prominent role in the
quantization of a classical system (cf. \cite{Isham1984,Souriau:69,Dowker1972, Leinaas1977, Woodhouse:80,
Horvathy:89, Morandi:92}   and references therein). More generally, topological ideas -not necessarily
related to quantization-  have permeated many branches of physics, including condensed matter physics and
quantum field theory.
 Thus, many of the nonstandard approaches to spin-statistics, rather than seeking a simple explanation for
 the connection, aim at a better understanding of it, stemming (at least partially)
  from the properties of the configuration space. This idea  goes back, at least, to the works of
  Finkelstein and Rubinstein \cite{Finkelstein1968}, Laidlaw and DeWitt \cite{Laidlaw1971},  and Leinaas and
   Myrheim \cite{Leinaas1977}. Many ideas inspired by those works have been put forward in the last decades.
   Perhaps the reader would agree that none of them can claim to have led to a substantially deeper
   understanding of spin-statistics. Why, then, the increasing interest in the topic? Indeed,
   the vast literature dealing with the spin-statistics connection (and particularly with attempts to find new
     and simpler ways to understand it) gives a hint as to the feeling of discomfort of many authors regarding
     its current physical status.

In my opinion, the current interest comes from different sources:
\begin{itemize}
\item    The need to re-examine the connection between spin and statistics in view of new developments in theoretical
   physics as, for instance, in the context of higher dimensional theories \cite{Boya2007},
quantum gravity \cite{Dowker1998,Alexanian2002}, algebraic quantum
   field theory \cite{Tscheuschner1989} or  non-commutative
    quantum field theory \cite{Balachandran2006}.
\item     The fact that many quantum phenomena that depend crucially on the spin-statistics connection take
place outside the realm of relativistic quantum field theory leads one to wonder whether there is a possible
explanation for it based solely on the principles of quantum mechanics.
 \item The current experimental search for violations of the exclusion principle: The recent Trieste conference on
 theoretical and experimental aspects of the spin-statistics connection showed that there is a real interest
 from both experimentalists and theoreticians in deepening our understanding of the connection and that
 intense experimental research around the topic is starting to  be done \cite{Milotti2010}.
\end{itemize}

 Let us finish this introduction with a description of the rest of the paper.
 Section \ref{sec:2} contains the proofs of the theorems
 announced in
 \cite{Papadopoulos2004} and \cite{Papadopoulos2010}. These theorems, along with their proofs,
  constitute the main contribution of the present paper. Section \ref{sec:3} presents  two applications of
  the results from section \ref{sec:2} to the problem of quantum indistinguishability.
 Here, our previous analysis of the Berry-Robbins approach is briefly reviewed.
In section \ref{sec:4} we present our conclusions and an outlook for
future work. In \ref{ap:A} we explain, for the benefit of the reader unfamiliar with these tools,  the
equivalence between projective modules and vector bundles.  Finally, \ref{ap:B} contains background material
on group representations that the reader might find useful in order to understand the motivation behind the
constructions presented in section \ref{sec:2}.
\section{An algebraic characterization of flat bundles on spaces with a finite fundamental group}
\label{sec:2} Let $\mathcal Q$ be a compact  manifold with  fundamental group $\pi_1(\mathcal Q)$. Let $C(\widetilde {\mathcal Q})$ denote the algebra of continuous\footnote{Here we will consider spaces of continuous functions (or sections) but in the applications we have in mind, it is the Hilbert space completion with respect
to a suitable inner product what we are interested in.}, complex functions on $\widetilde
{\mathcal Q}$, the universal covering space of $\mathcal Q$. The (right) action of $\pi_1(\mathcal Q)$ on $\widetilde
{\mathcal Q}$ induces, in a natural way, a representation of $\pi_1(\mathcal Q)$ on $C(\widetilde {\mathcal Q})$. On the
other hand, the homeomorphism $\widetilde {\mathcal Q}/\pi_1(\mathcal Q)\cong {\mathcal Q}$ induces, at the level of
functions, an algebra isomorphism between $C({\mathcal Q})$ and $C(\widetilde {\mathcal Q})^{\pi_1(\mathcal Q)}$, the
latter being the subalgebra of $C(\widetilde {\mathcal Q})$ consisting of all $\pi_1(\mathcal Q)$-invariant functions.
These facts allow one to regard $C(\widetilde {\mathcal Q})$ as a module over $C({\mathcal Q})$. The main
purpose of this section is to obtain an algebraic characterization of  the set of flat complex vector bundles
over $\mathcal Q$, when $\pi_1(\mathcal Q)$ is a finite group, through a detailed study of the module structure of
$C({\mathcal Q})$. The relevance of this characterization for the study of quantum indistinguishability has
been discussed at length in references \cite{Papadopoulos2004,Papadopoulos2010}. Let us remark that theorem 2.1 , lemma 2.2, lemma 2.3 and theorem 2.10 constitute the original contribution of the present paper, whereas theorems 2.5 and 2.7 are well-known results that have been included for the sake of completeness.
\subsection{Decomposition of $C(\widetilde {\mathcal Q})$ into a direct sum of $C({\mathcal Q})$-submodules}
Let $G$ be a finite group. A well-known result from representation theory states that every irreducible
representation of $G$ appears in the regular representation, with a multiplicity equal to its dimension. It
is then natural to consider the projection maps that single out each one of these representations. From the
explicit form of the projection operators, one can easily conclude that the representation space $\mathcal
F(G)$ of the regular representation (the vector space of complex valued functions on $G$) has a basis
consisting of all matrix elements of the irreducible representations of $G$. The generalization of this fact
to compact Lie groups is the celebrated Peter-Weyl theorem~\cite{Simon1996}. In this case, the representation
space is $L^2(G)$. Something similar occurs when we have an action $l: G\times M \longrightarrow M$ of a
group $G$ on a space $M$. Such an action induces, in a natural way, a linear action of $G$ on the space of
complex functions on $M$: If $f$ is such a function and $g\in G$, then $g\cdot f (m):= f(l_{g^{-1}}(m))$. In
the two cases mentioned above, the space $M$ is the group itself and the action $l$ is given by the group
multiplication. But we can also let $M$ be a more general space, on which $G$ acts. If $M$ is a compact
manifold and $G$ is a finite group acting freely on $M$, then the quotient space $M/G$ is again a manifold.
In this case, the  action $l$ also induces a representation of $G$ on the  space $C(M)$ of complex continuos
functions on $M$. The interesting point is that the decomposition of $C(M)$ into irreducible representations,
which is achieved purely by algebraic means, has a geometric interpretation, in the sense that it also
induces a decomposition  of $C(M)$ into a sum of finitely generated, projective $C(M/G)$-modules. It then
follows, from the Serre-Swan theorem, that each such subspace of $C(M)$  can be regarded as the space of
sections of some vector bundle over $M/G$.
 Therefore, we shall now seek to establish the following result:
\begin{thm}
\label{thm:1} Let $G$ be a finite group acting freely on the left on a compact manifold $M$ by $l:G\times M
\rightarrow M$. Then, the decomposition of $C(M)$ into irreducible representations induces, simultaneously, a
decomposition of $C(M)$ as a direct sum of finitely generated and projective $C(M/G)$-modules.
\end{thm}
We will divide the proof of this theorem in two parts (lemmas \ref{lema:A} and \ref{lema:B} below). Let
$C^{\mathrm G}(M)$ be the subspace of $C(M)$ consisting of all $G$-invariant functions:
\begin{equation}
C^{\mathrm G}(M):= \lbrace f\in C(M)\,|\, f(l_g(x))= f(x)\, \forall x\in M\rbrace.
\end{equation}
Since the product of two invariant functions is again an invariant function, this subspace is also an algebra
(and, in particular,  a ring). We can, therefore, regard $C(M)$ as a \emph{module} over $C^{\mathrm G}(M)$.
Moreover, since $C^{\mathrm G}(M)$ and $C(M/G)$ are isomorphic algebras, we can also regard $C(M)$ as a
$C(M/G)$-module. Let now $\widehat G$ denote the set of equivalence classes of irreducible representations of
$G$. We will tacitly identify each class $[R]$ in $\widehat G$ with a concrete representative $R\in [R]$,
chosen to be also a unitary representation. The dimension of the representation $R$ will be denoted by $n_R$.
Define now, for each $R\in \widehat G$ and for each pair of indices $i,j\in$ $\{1,\ldots,n_R\}$, a map
$E^R_{ij} : C(M) \rightarrow C(M)$, as follows:
\begin{equation}
\label{eq:Proj-E} E^R_{ij}f(x):= \frac{n_R}{|G|}\sum_{g\in G}R_{ij}(g^{-1})f(l_{g^{-1}}(x)),\hspace{1cm}(f\in
C(M),\, x\in M).
\end{equation}
Note the obvious similarity between this definition and the one given in (\ref{eq:B.3}). Functions in the
image of each $E^R_{ij}$ (for $R$ fixed) are related through the action of $G$ on $C(M)$. Indeed, for $f\in
C(M)$ and $g\in G$, we have: \numparts
\begin{eqnarray}
g\cdot(E^R_{ij}f)&=& \sum_{k=1}^{n_R}R_{kj}(g) (E^R_{ik}f),\\
E^R_{ij}(g\cdot f)&=& \sum_{k=1}^{n_R}R_{ik}(g) (E^R_{kj}f).\label{eq:2.3}
\end{eqnarray}
\endnumparts
 The maps $E^R_{ij}$ obey, furthermore, the following important orthogonality relations:
\begin{equation}
\label{eq:2.1} E^{R}_{ik}E^{R'}_{mn}=\delta_{R,R'}\delta_{i,n}E^{R}_{mk}.
\end{equation}
This follows directly from the orthogonality relations for the matrix elements of the representations. Taking
$R=R'$ and $i=k=m=n$  in \eref{eq:2.1}, we see that the operator $E^R_{ii}$ is a \emph{projection} operator.
Moreover, summing  over all $R$ in $\widehat G$ and (for each $R$) over all $i$, we obtain, as in
\eref{eq:B.4}, a completeness relation ($f\in C(M)$):
\begin{equation}
\label{eq:2.5} f=\sum_{R \in \widehat G} \sum_{i=1}^{n_R}E^R_{ii}f.
\end{equation}
In fact, we have:
\begin{eqnarray}
\sum_{R\in  \widehat G} \sum_{i=1}^{n_R}E^R_{ii}f(x) &=& \sum_{R\in \widehat G}
\sum_{i=1}^{n_R}\left\lbrack\frac{n_R}{|G|} \sum_{g\in G} R_{ii}(g^{-1})f(l_{g^{-1}}(x))\right\rbrack\nonumber\\
 &= & \sum_{R \in \widehat G}\sum_{g\in G}\frac{n_R}{|G|}\chi^R(g^{-1})f(l_{g^{-1}} x)\nonumber\\
& = &\sum_{g\in G}\frac{1}{|G|}\underbrace{\sum_{R\in \widehat G} n_R\chi^R(g^{-1}) }_{=|G|\delta_{e,g}}
f(l_{g^{-1}} x)\nonumber\\
& = & f(x).
\end{eqnarray}
\begin{lem}
\label{lema:A} Let $R\in \widehat G$ and $i\in \lbrace 1,\ldots,n_R\rbrace$ and define $\mathcal
E^R_i:=E^R_{ii}(C(M))$.
 Then $\mathcal E^R_i$ is a finitely generated $C(M/G)$-module. Moreover, $C(M)$ can be written, as a
 $C(M/G)$-module, as the direct sum
 \begin{equation}\label{eq:2.9}
 C(M)=\bigoplus_{R\in \widehat G}\bigoplus_{i=1}^{n_R}\mathcal E^R_i.
 \end{equation}
\end{lem}
\begin{proof}
Let $q:M\rightarrow M/G$ denote the quotient map induced by the (smooth) action $l:G\times M\rightarrow M$.
Since $G$ is finite and the action $l$ is free, it follows that $l$ is properly discontinuous without fixed
points, so the standard theorem \cite{Boothby:02} for quotients of such actions applies, showing that $M/G$
is also a smooth manifold. Following the proof of that theorem, one sees that it is possible to choose charts
for $M$ of the form $\lbrace(\tilde U_{\alpha, g}, \tilde \varphi_{\alpha,g})\rbrace_{\alpha\in I, g\in G}$,
with $I$ finite ($M$ is assumed to be compact) and with the following properties: \numparts
\begin{eqnarray}
  \mbox{For each $\alpha$ in $I$,
    $\tilde U_{\alpha,g}=g \cdot (\tilde U_{\alpha,e})$ for all $g\in G$}.\\
\mbox{For each $\alpha$ in $I$, $\tilde U_{\alpha,g}\cap \tilde U_{\alpha,g'}=\emptyset$ whenever $g\neq g'$;
$g,g'\in G$}.\label{eq:2.8}
\end{eqnarray}
\endnumparts
Here, the notation $l_g(x) =g\cdot x$ for the action has been used. It follows immediately from these
properties that, for a given $\alpha$,  $q(\tilde U_{\alpha,g})= q(\tilde U_{\alpha,e})$ holds,
 for all $g$.
Put now $U_\alpha:=q(\tilde U_{\alpha,g})$, with $g\in G$ arbitrarily chosen. The collection $\lbrace
U_\alpha\rbrace_{\alpha\in I}$ gives an open cover for $M/G$, with the following property:
\begin{equation}
\label{eq:q^-1(U_alpha)}
 q^{-1}(U_\alpha)=\bigcup_{g\in G}\tilde U_{\alpha,g}
                  =\bigcup_{g\in G }g\cdot (\tilde U_{\alpha,e}),
\end{equation}
i.e. the inverse image of $U_\alpha$ is a union of pairwise disjoint neighborhoods. Let now
$\{\phi_\alpha\}_{\alpha\in I}$ be a partition of unity for $M/G$, subordinated to the cover
$\{U_\alpha\}_{\alpha\in I}$
($\hbox{Supp}\,\phi_\alpha \subset U_\alpha\subset M/G$ and with the convention that $\sum_{\alpha \in
I}\phi_\alpha^{\,2}=1$ ).

 Each $\phi_\alpha$ gives place to a pull-back function
$q^*\phi_\alpha\in C^\infty(M)$, defined through $q^*\phi_\alpha:=\phi_\alpha\comp q$. This, in turn, can be
used to define, for each $\alpha$, the following map:
\begin{eqnarray}
\psi_{\alpha}(x)\,:=\,\left\lbrace
\begin{array}{c c}
     q^*\phi_\alpha(x), & \mbox{if}\;\; x\in \tilde U_{\alpha,e}\\
                    0, & \mbox{otherwise}.
\end{array}
\right.
\end{eqnarray}
Notice that these are indeed \emph{smooth} functions,  because of (\ref{eq:2.8}) and the fact that
 $\hbox{Supp} (q^*\phi_\alpha)= q^{-1}(\hbox{Supp}\, \phi_\alpha)\subseteq \bigcup_{g\in G}\tilde U_{\alpha,g}$.
More generally, let us define, for each $f\in C(M)$, the following (continuous) functions ($\alpha\in I, g\in
G$):
\begin{eqnarray}
f_{\alpha,g}(x)\,:=\,\left\lbrace
\begin{array}{c c}
     ((q^*\phi_\alpha) f)(x), & \mbox{if}\;\; x\in \tilde U_{\alpha,g}\\
                    0, & \mbox{otherwise}.
\end{array}
\right.
\end{eqnarray}

Define now, for $f\in C(M)$, the \emph{invariant} functions ($x\in M$)
\begin{equation}
f^{\mathrm{s}}_{\alpha,g}(x):=\sum_{h\in G} f_{\alpha,g}(h^{-1}\cdot x).
\end{equation}
From these definitions we immediately obtain the following identity:
\begin{equation}
\label{eq:ident-f^s} \sum_{g\in G}f^{\mathrm{s}}_{\alpha,g}(x)\psi_{\alpha}(g^{-1}\cdot x) =
\phi_\alpha([x])^2 f(x).
\end{equation}
Summing over $\alpha$, we then get
\begin{equation}
\label{eq:2.2} f(x)=\sum_{\alpha \in I}\phi_\alpha([x])^2 f(x)=\sum_{\alpha\in I}\sum_{g\in
G}f^{\mathrm{s}}_{\alpha,g}(x)\psi_{\alpha}(g^{-1}\tcdot x).
\end{equation}
Assume now that $f\in \mathcal E^R_i$, i.e. that $f=E_{ii}^R f $. Then, using \eref{eq:Proj-E}, \eref{eq:2.2}
and \eref{eq:2.3}, we obtain:
\begin{eqnarray}
f(x) & = & E_{ii}^R f(x)= \Bigl[E_{ii}^R\Bigl(\sum_{\alpha\in I}\sum_{g\in G}
f^{\mathrm{s}}_{\alpha,g}\,g\hspace{-0.05cm}\cdot \hspace{-0.05cm}\psi_{\alpha} \Bigr)\Bigr](x)\nonumber\\
& = &
 \Bigl[\sum_{\alpha\in I}\sum_{g\in G}E_{ii}^R\left(
f^{\mathrm{s}}_{\alpha,g}\,g\tcdot \psi_{\alpha} \right)\Bigr](x)\nonumber\\
&=&
 \sum_{\alpha\in I}\sum_{g\in G}f^{\mathrm s}_{\alpha,g}(x)
\bigl[E_{ii}^R (g\tcdot \psi_{\alpha}) \bigr](x)\nonumber\\
&=&
 \sum_{\alpha\in I}\sum_{g\in G}f^{\mathrm s}_{\alpha,g}(x)
\sum_{k=1}^{n_R}R_{ik}(g) \bigl[ E_{ki}^R \psi_{\alpha} \bigr](x)\nonumber\\
&=&
 \sum_{\alpha\in I}\sum_{k=1}^{n_R}\biggl(\sum_{g\in G}f^{\mathrm s}_{\alpha,g}(x)R_{ik}(g)
 \biggr)\bigl[ E_{ki}^R \psi_{\alpha} \bigr](x).\label{eq:2.4}
\end{eqnarray}
Define now $\psi^R_{\alpha, i k }:= E_{ki}^R \psi_{\alpha}$ (the reason for inverting the order of the
indices will become clear below) and observe that
\begin{itemize}
\item $f^{\mathrm s}_{\alpha,g}$ is invariant under $G$. Therefore, $\sum_{g\in G}f^{\mathrm
s}_{\alpha,g}R_{ik}(g)\in C^{G}(M)\cong C(M/G)$.
\item For each $\alpha\in I$ and each $k\in \lbrace 1,\ldots,n_R\rbrace$, the function $\psi^R_{\alpha, i k }$ belongs
 to $\mathcal E^R_i$. This follows directly from \eref{eq:2.1}.
\end{itemize}
We have therefore proved that $\mathcal E^R_i$ is a finitely generated $C(M/G)$-module since, as
 \eref{eq:2.4} shows, $\{\psi^R_{\alpha, i k}\}_{\alpha\in I,1\leq k\leq n_R}$ constitutes a \emph{finite}
  set of generators for
$\mathcal{E}^R_i$ as a module over $C^G(M)\cong C(M/G)$. The fact that $C(M)=\bigoplus_{R\in \widehat
G}\bigoplus_{i=1}^{n_R}\mathcal E^R_i$ follows directly from \eref{eq:2.5}.
\end{proof}

\begin{lem}
\label{lema:B}  There is, for every $R\in\widehat G$, an integer $N_R$ and a projector $P^R$ such that
$P^R(C(M/G)^{N_R})\cong \mathcal{E}^R_{i}$ for every $i\in{1,\ldots,n_R}$.
\end{lem}
\begin{proof}
Keeping the same conventions and notation as in the previous lemma, set $N_R=n_R |\,I|$ and  define a linear
map $P^R:C(M)^{N_R}\rightarrow C(M)^{N_R}$ as follows. The free module $C(M)^{N_R}$ is a direct sum of $N_R$
copies of $C(M)$. Therefore, every element $F$ of $C(M)^{N_R}$ can be written as an $N_R$-tuple
$F\equiv(F_{\alpha,j})$, with $\alpha\in I$ and $1\leq j\leq n_R$. Define $P^R F$ as the $N_R$-tuple whose
$(\alpha,j)$ component is given by
\begin{equation}
\label{eq:2.7} (P^R F)_{\alpha,j}:=\frac{|G|^2}{(n_R)^2}\sum_{\beta\in
I}\sum_{k,l=1}^{n_R}\overline{\psi^R_{\alpha, l j}(x)} \psi^R_{\beta, l k}(x) F_{\beta, k}(x),
\end{equation}
with the bar denoting complex conjugation. It is useful to regard $P^R$ as a $C(M)$-valued $N_R\times N_R$
(block) matrix.  Doing so, one sees that  the component $(j,k)$ corresponding to the block $(\alpha,\beta)$
must be given by
\begin{equation}
(P^R_{\alpha \beta}(x))_{jk}= \frac{|G|^2}{(n_R)^2}\sum_{l=1}^{n_R}\overline{\psi^R_{\alpha, l j}(x)}
\psi^R_{\beta, l k}(x).
\end{equation}
A straightforward  calculation shows that $(P^R)^2=P^R$. This means that $P^R(C(M)^{N_R})$ is a projective
module over $C(M)$. But it turns out that,  for any $g\in G$, we have $(P^R_{\alpha \beta}(g\tcdot x))_{jk}
=(P^R_{\alpha \beta}(x))_{jk}$. Therefore, we can use the isomorphism $C^G(M)\cong C(M/G)$ in order to obtain
a projective module over $C(M/G)$, by letting  $P^R$ act on $C(M/G)^{N_R}$. As shown below, this module is
isomorphic to $\mathcal E^R_i$. First let us remark that the module $P^R(C(M/G)^{N_R})$ is generated by the
columns of $P^R$, considered as $C(M/G)$-valued vectors. Let us now map the column of $P^R$ that is labeled
by the indices $(\alpha,k)$ to the generator $\psi^R_{\alpha,ik}$ of $\mathcal E^R_i$. Since both modules are
projective, it is not enough to give a bijection between the sets of generators to obtain a module
isomorphism: Both sets of generators must satisfy the same relations. The relations satisfied by the
generators of $P^R(C(M/G)^{N_R})$ are obtained from the condition $(P^R)^2=P^R$. On the other hand, using
\eref{eq:2.1} we obtain the following identity:
\begin{equation}
\label{eq:2.6} \sum_{\alpha\in I}\sum_{l=1}^{n_R}\psi^R_{\alpha,k'l}
\overline{\psi^R_{\alpha,kl}}=\frac{n_R^2}{|G|^2}\delta_{k',k}.
\end{equation}
Using this identity, together with \eref{eq:2.7}, we get:
\begin{equation}
\sum_{\alpha\in I}\sum_{k=1}^{n_R}\big(P^R_{\alpha\beta}([x]) \big)_{kj}\overline{\psi^R_{\alpha,ik}(x)} =
\overline{\psi^R_{\beta,ij}(x)}.
\end{equation}
This means that the generators of $\mathcal E^R_i$ satisfy the same relations as the columns of $P^R$.
Therefore, the $C(M/G)$-linear extension of the map used to identify the generators gives the desired
isomorphism.
\end{proof}
Let us make a few remarks about the meaning of this result. Consider a principal bundle $(G,P,M)$ with total
space $P$, structure group $G$ and base space $M$. Let $\xi= P\times_G V$ denote a vector bundle with fibre
$V$, associated to $P$ through some representation of $G$ on $V$. There is a well known theorem according  to
which every section $s$ of $\xi$ can represented by a vector valued function $\chi_s:P\rightarrow V$ which is
equivariant with respect to the $G$-actions on $P$ and $V$. Conversely, every $G$-equivariant map from $P$ to
$V$ represents some section of the bundle $\xi$. The result we have obtained is similar to that theorem in
the sense that the triple $(G, M, M/G)$ can be regarded as a principal bundle. Sections of a vector bundle
associated to this principal bundle through a representation $R$ of $G$ on  a vector space $V$ can therefore
be regarded as equivariant maps from $M$ to $V$. If $\dim V>1$, these maps will necessarily be vector valued.
What theorem \ref{thm:1} states is that these \emph{same} sections can in fact be represented by
\emph{scalar} (i.e. complex valued) functions on $M$. At first sight, it might appear as counterintuitive to
claim that a section of a bundle of rank higher than one over $M/G$ can be represented by a scalar function
on $M$, but as the proof of the theorem shows, this is indeed the case. According to a theorem of Milnor
\cite{Milnor1957}, every flat vector bundle over $M/G$ is associated to the principal bundle $(G, M, M/G)$
through some representation of $G$. From this point of view, theorem \ref{thm:1} says that \emph{every}
section of any flat vector bundle over $M/G$ can be represented by  a complex function on $M$. The fact that
these bundles are flat is reflected in the fact that the connections become, basically, the exterior
derivative on $M$. From a practical point of view, this could be a very convenient way of working with such
flat bundles. Potential applications of these techniques to quantization problems, Berry phases and spin
systems will be discussed in the last section.

\subsection{Invariant sections}
Pull-backs of vector bundles can be neatly expressed in terms of tensor products between  modules of sections
and algebras of functions~\cite{Gracia-Bond'ia2001}. This algebraic description  takes a particularly
interesting form when applied to pull-backs of quotients by finite groups, a form that turns out to be
essential for discussions of quantum indistinguishability~\cite{Papadopoulos2004,Papadopoulos2010}. Let us
start this section  recalling some basic mathematical facts.

\begin{definition} A vector bundle  $\xi=(E(\xi),\pi,M)$ over the  $G$-space $M$ is called a $G$-bundle when the following
conditions hold:
\begin{itemize}
\item The total space $E(\xi)$ is itself a $G$-space  (the corresponding
action being denoted with $\tau$).
\item The projection $\pi$ is $G$-equivariant, i.e.
$\pi \comp \tau_g = \rho_g \comp \pi$ for all $g$ in $G$.

\item The restriction
 $\tau_g\big|_{\pi^{-1}(m)} :\pi^{-1}(m)\longrightarrow\pi^{-1}(g\cdot m)$
of the action to the fibers is a vector space isomorphism.
\end{itemize}
A morphism between two $G$-bundles is a $G$-equivariant bundle morphism. The notation $\xi_1 \cong_G \xi_2$
will be used whenever $\xi_1$ and
 $\xi_2$ are equivalent as  $G$-bundles.
\end{definition}
For $G$ finite, we have:
\begin{thm}[cf.\cite{Atiyah1967}]
\label{prop:Atiyah} If $M$ is   $G$-free, there is a bijective correspondence between
 $G$-bundles over  $M$ and  bundles over $M/G$ by $\eta\rightarrow \eta/G$.
\end{thm}
The precise meaning of this theorem is the following. Let
 $q:M\rightarrow M/G$ denote the canonical projection and let $\xi$ be a vector bundle over $M/G$.
  Then, the pull-back $q^*$ induces -in a natural way- an action $\tau_\xi$ of $G$ on $E(q^*\xi)$, given by
\begin{equation}
\label{eq:g(m,y)} \tau_\xi(g,(m,y))=g\cdot(m,y):=(g\cdot m,y),\;\; g\in G,\;\; (m,y)\in E(q^*\xi)
\end{equation}
This action is also free, implying that the quotient $E(q^*\xi)/G$ is also a manifold. One then shows that
$E(q^*\xi)/G$ is the total space of a vector bundle (denoted $q^*\xi/G$) over $M/G$, which is isomorphic to
$\xi$: $q^*\xi/G\cong\xi$. On the other hand, let $\eta$ be a $G$-bundle over $M$. The quotient $E(\eta)/G$
is -again- the total space of a bundle, but this time  over $M/G$. Its  pull-back turns out to be
$G$-isomorphic to $\eta$: $q^*(\eta/G)\cong_G \eta$.

Consider now a continuous map $\phi: M\rightarrow N$ and a vector bundle $\xi$ over $N$.
Noticing that $\phi$ induces a ring homomorphism $\phi^*:C(N)\rightarrow  C(M)$ (through $f \mapsto
\phi^*f:=f\comp\phi$), we obtain a $C(N)$-module structure on $\Gamma(\phi^*\xi)$:
\begin{eqnarray}\label{eq:1}
C(N)\times \Gamma(\phi^*\xi) & \longrightarrow &   \Gamma(\phi^*\xi)\\
\hspace{0.6cm}(\,f\,,\,s\,)& \longmapsto & f\cdot s:= (\phi^*f)s.\nonumber
\end{eqnarray}
If for any given section $\sigma\in \Gamma(\xi)$ we define a new one  by
\begin{eqnarray}
\label{eq:2}
\phi^* \sigma:M & \longrightarrow & \hspace{0.5cm}E(\phi^*\xi)\\
       \hspace{1.1cm}     x & \longmapsto     & \;\;(x,\sigma\comp\phi(x)),\nonumber
\end{eqnarray}
then we obtain the following  homomorphism of $C(N)$-modules:
\begin{eqnarray}
\label{eq:F}
F^\phi:\Gamma(\xi)& \longrightarrow & \Gamma(\phi^*\xi)\\
   \hspace{1.1cm}\sigma & \longmapsto & F^\phi(\sigma)\equiv\phi^*\sigma. \nonumber
\end{eqnarray}
$F^\phi$ is clearly a $C(N)$-linear map:
\[F^\phi(f\cdot\sigma)=\phi^*(f\cdot\sigma)\stackrel{(\ref{eq:2})}=
(f\comp\phi)\phi^*\sigma
  \stackrel{(\ref{eq:1})}=f\cdot \phi^*\sigma=f\cdot F^\phi(\sigma) .\]
\begin{remark}
\label{rem:generadores} It is important to remark that $F^\phi$ is not, in general, an isomorphism. Indeed,
although we may choose generators  for $\Gamma(\phi^*\xi)$ of the form $\sigma^\prime_i = F^\phi(\sigma_i)$,
we see from $\mbox{Im}(F^\phi)=\{\sum_i(f_i\mbox{\small$\comp$}\phi)\sigma^\prime_i \,|\, f_i\in C(N)\}$ that
$F^\phi$ is not surjective in general, because the elements in the image module are only linear combinations
of the generators over the \emph{subspace} $\phi^*(C(N))$ of $C(M)$. In other words: In order to remain
inside $\mbox{Im}(F^\phi)$, the generators $\sigma^\prime_i$ may  be multiplied only by elements of $C(N)$.
\end{remark}
The previous remark suggests that we might obtain $\Gamma(\phi^*\xi)$ from $\mbox{Im}(F^\phi)$ if we are able
to replace, somehow, $C(N)$ by $C(M)$. This change of ring can in fact be performed, with the help of the
tensor product, because both $C(M)$ and $\Gamma (\xi)$ can be regarded as $C(N)$-modules. One can therefore
define the following homomorphism of $C(M)$-modules:
\begin{eqnarray}
\label{eq:3}
\Phi: C(M)\otimes_{C(N)}\Gamma(\xi) & \longrightarrow &  \hspace{0.3cm}\Gamma(\phi^*\xi)\\
\hspace{0.90cm}\sum_k a_k\otimes\sigma_k \hspace{1cm}          & \longmapsto & \sum_k a_k
F^\phi(\sigma_k).\nonumber
\end{eqnarray}
This is the key idea behind the following  (well-known) result.
\begin{thm}[cf.~\cite{Gracia-Bond'ia2001,Madsen1997}]
\label{thm:pull-back-alg} The map defined through (\ref{eq:3}) provides an isomorphism\linebreak
 $C(M)\otimes_{C(N)}\Gamma(\xi)\cong \Gamma(\phi^*\xi)$ of $C(M)$-modules.
\end{thm}

If -for the situation considered  in the previous paragraphs- we set $N\equiv M/G$ and $\phi\equiv
q:M\rightarrow M/G$ in theorem  \ref{thm:pull-back-alg}, then we can construct an injective $C(M/G)$-module
homomorphism $\Phi^G:\Gamma(\xi)\hookrightarrow \Gamma(q^*\xi)$, as follows.  The decomposition
(\ref{eq:2.9}) of $C(M)$ into $C(M/G)$-submodules allows us to write  $C(M)$  in the form $C(M)=C^G(M)\oplus
\mathcal{E}$, where as already remarked  $C^G(M)\cong C(M/G)$ holds,  and with $\mathcal{E}$ having the
structure of a projective $C(M/G)$-module. Hence  we obtain:
\begin{equation}
\label{eq:2.10}
 C(M)\otimes_{C(M/G)}\Gamma(\xi) \cong \Gamma(\xi)\oplus
                    \left( \mathcal{E}\otimes_{C(M/G)}\Gamma(\xi)\right).
\end{equation}
Denote with $i:\Gamma(\xi)\hookrightarrow C(M)\otimes_{C(M/G)}\Gamma(\xi)$  the inclusion induced by
(\ref{eq:2.10}). Making use of theorem \ref{thm:pull-back-alg} we   can define $\Phi^G:=\Phi\comp\, i$ and in
that way obtain the desired result.
\begin{remark}
\label{rem:remark-importante} It is important to notice that, although $\Phi^G(\Gamma(\xi))$ and
$\Gamma(\xi)$ are isomorphic as $C(M/G)$-modules, $\Phi^G(\Gamma(\xi))$ is actually contained in
$\Gamma(q^*\xi)$. This means that, although every section from $\Gamma(\xi)$ can be ``replaced'' by one from
$\Phi^G(\Gamma(\xi))$, sections from $\Phi^G(\Gamma(\xi))$ may \emph{only} be multiplied by functions in
$C^G(M)$, if we want to identify $\Phi^G(\Gamma(\xi))$ and $\Gamma(\xi)$ as \emph{modules}.
\end{remark}
\begin{remark}
From  (\ref{eq:3}) and (\ref{eq:2.10}) we see that
$\Phi^G(\sigma)=F^q(\sigma)$.
\end{remark}
It is possible to give a more explicit description of the image of $\Phi^G$. In fact, one finds that
$\Phi^G(\Gamma(\xi))$ equals the space of invariant sections of the pull-back bundle. The next result is very
important in connection to our discussion of the single-valuedness condition of Berry and Robbins
(cf.~\cite{Papadopoulos2004,Papadopoulos2010}).
\begin{thm}
\label{thm:inv-sections} $\Gamma(\xi)\cong\Phi^G(\Gamma(\xi))=\Gamma^{G}(q^*\xi)
                  :=\{s\in \Gamma(q^*\xi) \;|\; g\cdot s = s
                                 \;\,\forall\,g\in \Gamma\}.$
\end{thm}
\begin{proof}[Proof]
The first equality is clear, since $\Phi^G$ is an injective homomorphism. Every section in
$\Phi^G(\Gamma(\xi))$ is of the form $q^*\sigma$, with $\sigma\in \Gamma (\xi)$. From the definition of
$q^*\sigma$ (see Eq. (\ref{eq:2})) and from the form of the action $\tau$ induced induced on $q^*\xi$ by $q$
(see Eq.(\ref{eq:g(m,y)})) it follows that a section of the form $q^*\sigma$ is invariant:
\begin{eqnarray*}
\left( g\cdot q^*\sigma\right)(x)
              & = & \tau_g(g^{-1}\cdot x,\sigma\comp \,q(g^{-1}\cdot x))\\
            &\stackrel{(\ref{eq:g(m,y)})}{=}& (x,\sigma\comp\, q(g^{-1}\cdot x))\\
            & = & (x, \sigma\comp \,q(x)  )\\
            & = & q^*\sigma(x).
\end{eqnarray*}
Conversely, every invariant section must be of the form $q^*\sigma$: Given $s\in \Gamma(q^*\xi)$, there is a
continuous map $y:M \rightarrow E(\xi)$ with $\pi\comp y = q$ and $s:x\mapsto (x,y(x))$. If $s$ is invariant,
then  $y(g\cdot x)= y(x)$ holds for all $g$ in $G$, so that one can define  a section $\sigma\in\Gamma(\xi)$
through $\sigma([x]):=y(x)$, for which $s=q^*\sigma$ holds.
\end{proof}

\begin{remark} At this point it might not be very clear what kind of advantage can we obtain from describing all these geometric objects in terms of algebraic ones. In the end, as long as we stay within classical geometry, the theorems of Gelfand-Naimark and of Serre-Swan assert that there a complete equivalence, at the level of categories. There are two main reasons for insisting in this approach: (i) at some point one would like to make contact with the (more algebraic) language of quantum field theory, in order to try to interpret the spin-statistics theorem with the geometry of the configuration space (whether this is possible or not is a different issue) and (ii) we have already had the experience, working on a geometric interpretation of the Schwinger construction of Berry-Robbins, that the algebraic approach is really worth pursuing: In that case, the whole intricacy of the Schwinger construction reduces to the assertion that the spin basis is represented by the $C(\mathbb R P^2)$-projective module $A_+\oplus A_-^3$, where $A\pm$ stands for the space of continuous even/odd functions over the 2-sphere. As shown in \cite{Paschke2001,Papadopoulos2004}, this can be obtained from general symmetry considerations.
\end{remark}
\section{Applications to quantum indistinguishability}
\label{sec:3} The results presented in the previous section were originally motivated by our efforts to
understand the Berry-Robbins construction from a geometric perspective~\cite{Papadopoulos2004}. The  approach
we have chosen (fully exploiting the Serre-Swan equivalence) turns out to be very well suited for the study
of several problems in quantum mechanics where multiply-connected configuration (or parameter) spaces are
relevant. We will comment on various such possibilities in section \ref{sec:4}.

Here we will discuss an application to the problem of quantum indistinguishability. It goes directly to the core of the Berry-Robbins approach to spin-statistics.  Since this involves the use of the concept of single-valuedness of wave functions, it is convenient to add some comments on it. Let us define
 \[
 \tilde Q_N:=(\mathbb R^3\times\cdots \mathbb \times \mathbb R^3\setminus \Delta),
 \]
where $\Delta$ denotes the set of configurations with two or more coinciding particles. According to Leinaas and Myrheim \cite{Leinaas1977}, if we want to take into account the intrinsic quantum indistinguishability of the particles but still want to work with wave functions defined over a classical configuration space, we  have to consider the quotient space $Q_N:= \tilde Q_N/S_N$ (obtained through identification of all permuted configurations) as the physical  configuration space. This space is multiply connected, which implies there will be inequivalent quantizations of the same classical system. Each possible quantization of the system will be given by a Hilbert space that is to be obtained as the completion of a space of sections of a certain vector bundle over $Q_N$. The elements of that Hilbert space are representatives of rays in the projective Hilbert space and it is important to remark that we are free to multiply any representative by an arbitrary phase factor, without changing the ``single-valued'' character of the section chosen as representative.  In the present context of quantum indistinguishability, the term ``single-valued wave function'' refers to a section of that bundle. As we have seen in the previous section, it is possible to describe a section on a certain bundle over $Q_N$ as a function on the covering space $\tilde Q_N$. This function will in general acquire  different values at points that physically correspond to the same configurations of particles. For this reason such functions are sometimes called ``multiple-valued''. In any case, it should be clear that in all cases relevant to our discussion, including references \cite{Berry1997,Leinaas1977}, it is  $Q_N$ and not $\tilde Q_N$ the space that is to be considered the physical configuration space.

The main purpose of this section is to illustrate the results of section 2 by means of an application of direct physical interest. The main assertion, questioning the applicability of the single-valuedness condition in the form
advocated by Berry and Robbins, might appear as something difficult to accept. Therefore we urge the interested reader to carefully follow the argumentation and explicit computations presented in references \cite{Papadopoulos2004} and \cite{Papadopoulos2010} both of which are based on the mathematical tools developed in section 2.

Very roughly, the strategy of Berry and Robbins in their approach to quantum indistinguishability consists in
replacing the standard spin states by position-dependent ones, in such a way that an exchange of position
leads, at the same time, to an exchange in the spin degrees of freedom. Thus, instead of working with spin
states (in this section I will only be concerned with  the  2 particle case) of the form $|m,m'\rangle$, they work with spin states of the form $|m,m'(\bi
r)\rangle$, where $\bi r$ stands for the relative position vector of the two particles. These transported
spin states are required to satisfy a relation of the form $|m',m(-\bi r)\rangle=(-1)^K|m,m'(\bi r)\rangle$
and also to depend smoothly on $\bi{r}$ and to satisfy a certain property related to parallel transport. It
should be remarked that the existence of such a basis is not at all obvious. The construction presented by
Berry and Robbins in \cite{Berry1997} makes use of Schwinger's representation of spin and has the virtue of
providing a transported spin basis that works for any value of the spin, as well as the physically correct
statistics ($K=2S$). However, as already remarked, there are other possible constructions satisfying all the
requirements but giving the wrong statistics. On the other hand, a construction of a spin basis satisfying their three requirements of smoothness, spin exchange and parallel transport in the case of arbitrary $N$ turns out to be
very difficult to implement. This difficulty led the authors of \cite{Berry1997} to a mathematical problem that was later on reformulated and solved by Atiyah and collaborators \cite{Atiyah2001,Atiyah2002,Atiyah2002a}.   Harrison and Robbins \cite{Harrison2004}  used these results in order to study generalizations of the Schwinger construction for the case of general $N$.

Returning to the 2 particle case, in the
Berry-Robbins approach, wave functions are given by expressions of the form
\begin{equation}
\label{eq:3.1} |\Psi(\bi{r})\rangle=\sum_{m,m'}\psi_{m,m'}(\bi{r})|m,m'(\bi{r})\rangle,
\end{equation}
and are required to satisfy the following \emph{single-valuedness} condition:
\begin{equation}
\label{eq:3.2} |\Psi(\bi{r})\rangle=|\Psi(-\bi{r})\rangle.
\end{equation}
The purpose of this condition is to incorporate indistinguishability already at the level of configuration
space since, for indistinguishable particles, the points $\bi r$ and $-\bi r$ represent exactly the same
configuration. In the case of two particles one can therefore say that the \emph{physical} configuration
space is the projective plane\footnote{The configuration space is actually of the form $\mathbb R^3 \times
\mathbb R_+ \times \mathbb R P^2$, but we are ignoring the first two factors, corresponding to the center of
mass position vector and the relative distance between the particles, because they do not play any role in
the discussion.} $\mathbb R P^2$, obtained from the two-sphere $S^2$ through the identification $\bi{r}\sim
-\bi{r}$. Although the physical motivation for the imposition of the single-valuedness  requirement is very
clear, its implementation in the specific form  (\ref{eq:3.2}) is a very subtle matter. To bring it to the
point, and in order to illustrate the usefulness of the techniques presented in the previous section, let us
highlight the following point.
\begin{remark} The wave function $|\Psi(\bi{r})\rangle$ considered by Berry and Robbins is a map
$|\Psi(\,\cdot\,)\rangle: S^2 \rightarrow \mathbb C^k$, from the sphere to a vector space $\mathbb C^k$, with
the value of $k$ depending on the explicit construction of the transported spin basis. However,  for a
generic configuration space, the wave function will be a section of some vector bundle over the configuration
space. Since in the present case the physical configuration space is the projective space $\mathbb R P^2$,
the wave function should be given by a section of some vector bundle over $\mathbb R P^2$, not over the
sphere. For this reason, the wave function $|\Psi(\bi{r})\rangle$ of Berry and Robbins and, in particular,
the single-valuedness condition, must be treated with caution.
\end{remark}
Taking this remark into account, let us consider the situation of theorem \ref{thm:inv-sections}, with
$M=S^2$ and $G=\mathbb Z_2$. In that case we have $q:S^2\rightarrow S^2/\mathbb Z_2\cong\mathbb R P^2$. Now,
let us assume, according to the remark, that the wave function $\varphi$ is a section on a vector bundle
 $\xi$ over $\mathbb R P^2$, i.e., $\varphi\in \Gamma(\xi)$. By means of the map $\Phi^{\mathbb Z_2}$ from
 theorem \ref{thm:inv-sections} we then obtain an isomorphic copy of $\varphi$
 which is a section of the  vector bundle $q^*\xi$ (a bundle over $S^2$). According to theorem
 \ref{thm:inv-sections}, this new section has well-defined transformation properties under permutations of
 its argument: It must be \emph{invariant} with respect to the natural  $\mathbb Z_2$-action on
 $q^*\xi$ (cf. (\ref{eq:g(m,y)})). In order to obtain an explicit form of this invariance condition, it is
 convenient to regard $q^*\xi$ as a sub-bundle of a \emph{trivial} vector bundle $S^2\times \mathbb C^k$, for
 a  suitably chosen $k$ (this is always possible). After doing this, we find that there must be a map
 $|\bar\varphi(\,\cdot\,)\rangle:S^2\rightarrow \mathbb C^k$ such that
 $\Phi^{\mathbb Z_2}(\varphi)(\bi r)=(\bi r, |\bar \varphi(\bi r)\rangle)$. Therefore, if we want to keep the interpretation
 that (because of indistinguishability) the configuration space is $\mathbb R P^2$, we must conclude that
 this map is precisely the Berry-Robbins wave function: $|\bar \varphi(\bi r)\equiv |\Psi(\bi r)\rangle$. The
 consequences of this fact for the single-valuedness condition (\ref{eq:3.2}) are of the  utmost importance since,
 as shown in \cite{Papadopoulos2010} by means of an explicit example, it might well happen  that
$|\Psi(-\bi r)\rangle=-|\Psi(\bi r)\rangle$.

Therefore, although  the purpose of the single-valuedness
condition of Berry and Robbins is to incorporate
indistinguishability already at the level of configuration space,
the specific requirement (\ref{eq:3.2}) carries with itself
certain ambiguities. This has been analyzed in full detail in
\cite{Papadopoulos2004} and \cite{Papadopoulos2010}. Instead we
propose the imposition of invariance on the wave function, in the
sense of theorem \ref{thm:inv-sections}, as a more clear and
concise way to incorporate indistinguishability into the formalism
of quantum mechanics.
\section{Conclusions and outlook}
\label{sec:4}
In the present paper I have presented a mathematical framework for the study of quantum mechanics on spaces
with a finite fundamental group which is based on the Gelfand-Naimark and Serre-Swan equivalences. The
original motivation for this work has been the problem of quantum indistinguishability. An important aspect
of this problem has been discussed using the tools developed in this paper.

 Let us finish this paper with a brief
outline of some problems for which the results presented in
section \ref{sec:2} might be useful:
\begin{itemize}
\item Following the spirit of the Berry-Robbins approach, Peshkin
has argued in \cite{Peshkin2003} that spin zero particles must be
bosons. Although his work has been criticized by Allen and
Mondragon in \cite{Allen2003}, these authors do not provide a
conclusive argument against the possibility of a topological
origin of the spin-statistics connection. In his reply
\cite{Peshkin2003a} and  subsequent work \cite{Peshkin2006},
Peshkin still argues that, when taking indistinguishability into
account, spin zero fermions can be disregarded. It would be
interesting to spell out in detail Peshkin's calculations using
the language presented here. This could provide a conclusive
answer to that controversy.
\item As remarked in our previous papers \cite{ Reyes-Lega2010, Benavides2010}, Kuckert's work
\cite{Kuckert2004} suggests a close connection between the
spin-statistics connection, on one hand, and the angular momentum
algebra for one and two particles and their possible intertwining
relations, on the other. The techniques presented in this paper
could be used to recast Kuckert's proposal in global language. A
first step in this direction has been taken in
\cite{Reyes-Lega2010}, but a detailed analysis of the
representation theoretic aspects of the problem remains to be
done.
\item Recently, new applications of geometric phases and
topological invariants have appeared in the context of spin
systems and in relation to quantum phase transitions. In
particular, it has been shown  in \cite{Contreras2008} that it is
possible to relate certain topological invariants computed from the
Hamiltonian and relate them to quantum criticality. The use of
projective modules, as presented in the present paper, can be very
convenient in this context from the computational point of view.
\end{itemize}
\ack
It is a pleasure to thank M. Paschke for his key guidance and encouragement  during the initial phase of this work, and also to N.A. Papadopoulos
and F. Scheck for the many enlightening discussions we had over a period of several years about this work. My intellectual debt to them will remain unmatched.
Financial support from DAAD, Universidad de los Andes and Colciencias is gratefully acknowledged.
\appendix
\section{Projective modules and the Serre-Swan theorem}
\label{ap:A}
\renewcommand{\thethm}{A.\arabic{thm}}
\begin{definition} Let $R$ be a unital ring and $\mathcal E$  an abelian group. $\mathcal E$ is said to be a (left)
$R$-\emph{module} if there is a map $R \times \mathcal{E} \rightarrow  \mathcal{E}$, $(r,\varphi)  \mapsto
r\cdot \varphi$ such that, for all $r_1,r_2\in R$ and $\varphi,\varphi_1,\varphi_2\in \mathcal E$, the
following relations hold:
\begin{equation}
\eqalign{r\cdot(\varphi_1+\varphi_2) = r\cdot \varphi_1 + r\cdot \varphi_2 \cr
(r_1+r_2)\cdot\varphi = r_1\cdot \varphi + r_2\cdot \varphi \cr
(r_1\,r_2)\cdot\varphi = r_1\cdot(r_2 \cdot \varphi)\cr
 1\cdot\varphi= \varphi.}
\end{equation}
\end{definition}
\begin{definition} Let $\mathcal E$ be an $R$-module. A subset $S$ of $\mathcal E$ is said to be a set of
\emph{generators} for $\mathcal E$, if every $\varphi\in \mathcal E$ can be written as a sum
$\varphi=\sum_{\sigma\in S} r_\sigma \cdot\sigma$, where $r_\sigma\in R$ and $r_\sigma =0$ for all but a
finite number of elements $\sigma\in S$. If, in addition, the set $S$ is finite, we say that $\mathcal E$ is
\emph{finitely generated}. If $S$ is a set of generators for $\mathcal E$ and if for every $\varphi\in
\mathcal E$ the expansion $\varphi=\sum_{\sigma\in S} r_\sigma \cdot\sigma$ is unique, we call $S$ a
\emph{basis}. A module $\mathcal E$ is called \emph{free} if it admits a basis.
\end{definition}
Given two $R$-modules $\mathcal E$ and $\mathcal F$, we can construct their \emph{direct sum} $\mathcal
E\oplus \mathcal F$ as the set of pairs $(\eta,\varphi)$ with $\eta\in \mathcal E$ and $\varphi \in \mathcal
F$. Addition is defined componentwise, $(\eta,\varphi) +(\eta',\varphi')= (\eta+\eta',\varphi+\varphi')$ and
the $R$-operation  given by $r\cdot(\eta,\varphi)= (r\cdot \eta,r\cdot\varphi)$. Note that a ring $R$ can
also be considered as a module over itself. We can therefore construct the $n$-fold sum
$R^n=R\oplus\cdots\oplus R$, seen as a left $R$-module, with $R$-product
$r\cdot(r_1,\ldots,r_n)=(r\,r_1,\ldots, r\,r_n)$. This module is free, with a standard basis given by the $n$
elements of the form $(0,\ldots,0,1,0,\ldots,0)$.
\begin{proposition}[cf.\cite{Hilton1971}, Proposition 4.1] Let $\mathcal E$ be a finitely generated free $R$-module, with basis
$S=\lbrace\sigma_1,\ldots, \sigma_n\rbrace$. Then $\mathcal E$ is isomorphic to the $R$-module $R^n$.
\end{proposition}
Free modules are similar to vector spaces in the sense that if $S$ is a basis for $\mathcal E$, then
$\sum_{\sigma\in S} r_\sigma\cdot \sigma=0$ implies $r_\sigma=0$ for all $\sigma\in S$. But, of course, there
are modules which are not free.
\begin{example}
Let $\mathbb Z_n$ denote the additive group of integers modulo $n$. Its elements are equivalence classes
$[m]_n$, with $[m]_n=[m']_n$ if and only if $m=m' (\mbox{mod } n)$ and addition defined by
$[r]_n+[s]_n:=[r+s]_n$. $\mathbb Z_n$ can also be regarded as a ring, if we define the product by
$[r]_n[s]_n:=[rs]_n$. Regarded as a module over itself, $\mathbb Z_n$ is a free module, generated by $[1]_n$.
But we can also consider, say, $\mathbb Z_2$ as a $\mathbb Z_6$-module, by defining the ring operation
through $[m]_6\cdot[n]_2:= [mn]_2$. In that case, one easily checks that $\mathbb Z_2$ is generated (over
$\mathbb Z_6$) by $[1]_2$. But the set $S=\lbrace[1]_2\rbrace$ is \emph{not} a basis, because $[1]_2$
satisfies certain \emph{relations} as, for example, $[3]_6\cdot[1]_2=[1]_2$. This shows that $S$ is not
linearly independent or, in other words, that $\mathbb Z_2$ is \emph{not} a free $\mathbb Z_6$-module.
Analogous computations  show that $\mathbb Z_3$ is also a $\mathbb Z_6$-module which, again, is not free.
\end{example}
The following definition of projective module is not the most ``elegant'' one, but is well-suited for our
purposes.
\begin{definition} An $R$-module $\mathcal P$ is said to be \emph{projective} if it is a direct summand of a
free module.
\end{definition}
\begin{proposition}[cf.\cite{Gracia-Bond'ia2001}, Proposition 2.22]\label{ap:A-proposition}
An $R$-module $\mathcal P$ is projective if and only if it is of the form $p (\mathcal F)$, where  $\mathcal
F$ is a free $R$-module and $p: \mathcal F\rightarrow \mathcal F$ an $R$-module homomorphism that is an
idempotent, that is, such that $p^2=p$.
\end{proposition}
\begin{example}
From the previous example we know that both $\mathbb Z_2$ and $\mathbb Z_3$ can be regarded as $\mathbb
Z_6$-modules. It follows that $\mathbb Z_2\oplus \mathbb Z_3$ is also a $\mathbb Z_6$-module. In fact, it is
isomorphic to $\mathbb Z_6$. The isomorphism is given by the following map:
\begin{eqnarray}
\label{eq:A.Z2}
\;\;\mathbb Z_2\oplus \mathbb Z_3 &\longrightarrow &\;\;\;\;\; \mathbb Z_6 \nonumber\\
([m]_2,[n]_3)&\longmapsto & [3m+ 2n]_6.
\end{eqnarray}
Let us consider the $\mathbb Z_6$-linear map $p: \mathbb Z_6\rightarrow \mathbb Z_6$ defined by $p([m]_6):=
[3m]_6$. We have $p^2([m]_6)=[9m]_6= [3m + 6m ]_6= p ([m]_6)$, so that $p(\mathbb Z_6)$ is a projective
module, given precisely by the image of $\mathbb Z_2$ under \eref{eq:A.Z2}. We thus have $\mathbb Z_2\cong p
(\mathbb Z_6)$, illustrating the previous proposition.
\end{example}
\begin{example}
Let $M$ be a compact manifold and $E\stackrel{\pi}{\rightarrow} M$ a (real/complex) vector bundle. Then the
space $\Gamma(E)$ of all continuous sections of the bundle has the structure of a finitely generated
projective module over the algebra $C(M)$ of continuous (real/complex) functions on $M$. In fact, the space
$C(M)$, being an algebra, is also a ring, with respect to pointwise addition and multiplication of functions.
$\Gamma(E)$ is a vector space over $\mathbb K(=\mathbb R,\mathbb C)$ and, in particular, an abelian group.
The ring operation is given by the  product $C(M)\times \Gamma(E)\rightarrow \Gamma(E)$, $(f,\sigma)\mapsto
f\cdot \sigma$, pointwise defined: $(f\cdot\sigma)(x):= f(x) \sigma(x)$. The fact that this module is
finitely generated and projective is part of the \emph{Serre-Swan theorem}\cite{Serre1957-58,Swan1962}. This
theorem states that there is an equivalence between the categories of vector bundles over compact spaces $M$,
on one side, and of finitely generated projective modules over the respective algebras $C(M)$, on the other.
\end{example}
\begin{remark}\label{ap:A-remark} A detailed proof of this theorem can be found, e.g., in \cite{Gracia-Bond'ia2001}. For the
purpose of this paper, the following information about the construction of the equivalence will be
sufficient:
\begin{itemize}
\item A vector bundle with typical fibre $V$ can be completely described by an open cover
$\lbrace U_1,\ldots,U_m\rbrace$ of $M$, together with a family of transition functions $g_{ij}:U_i\cap
U_j\rightarrow \mbox{Gl}(V)$ satisfying certain compatibility (cocycle) conditions \cite{Nakahara1990}. Let
$\lbrace \psi_j \rbrace_j$ be a partition of unity subordinate to the cover $\lbrace U_j\rbrace_j$ (here we
use the convention $\sum_j\psi_j^2=1$). Set $N=m\dim (V)$. Then, the projective module corresponding to this
bundle is defined by an idempotent $P: C(M)^N\rightarrow C(M)^N$, acting  as an $N\times N$ matrix-valued
function constructed from $\dim (V)^2$ blocks, the $ij$-block of which is given  by $g_{ij}\psi_i\psi_j$.
\item On the other hand, it follows from proposition \ref{ap:A-proposition} that any finitely generated
projective module over $C(M)$ is given by a $C(M)$-valued matrix idempotent $P$ as $P(C(M)^N)$, for some $N$.
In this case, the fibre of the corresponding vector bundle over the point $x\in M$ is given by $P(x)(\mathbb
K^N)$, where $\mathbb K=\mathbb R,\mathbb C$, depending on whether we are working with real, respectively
complex functions.
\item A vector bundle of rank $r$ which is \emph{not} trivial will give place to a projective module which is \emph{not}
free. Thus, finding a set of $r$ linearly independent, nowhere-vanishing sections of such a bundle amounts,
at the algebraic level, to find a basis for the module.
\end{itemize}
\end{remark}
\begin{example}
The M\"{o}bius bundle is the simplest example  of a nontrivial \emph{real} line bundle.  Using the coordinate
$\theta$ to denote points in the circle  $S^1$ ($0\leq \theta\leq 2 \pi$), we can cover it with two open
neighbourhoods, the first one, $U_a$, including all points of $S^1$ but the one corresponding to  $\theta=0$
and the second one, $U_b$, including all points of $S^1$ but the one corresponding to $\theta = \pi$. Using
the bundle construction theorem, we can construct a real line bundle by specifying transition functions
$g_{ab}: U_a\cap U_b\rightarrow \mbox{Gl}(n,\mathbb{R})$ satisfying the usual cocycle conditions (in this
case $n=1$). Since $U_a\cap U_b$ has two connected components, we can define
\begin{equation}
\label{eq:A.1} \label{cases}
g_{ab}(\theta):=\cases{+1& for  $\;0 <\theta< \pi$,\\
-1& for  $\;\pi <\theta< 2\pi$.\\}
\end{equation}
Putting  $g_{ba}= g_{ab}$ and $g_{aa}=g_{bb}=1$, we obtain  transition functions that are locally constant
(hence continuous) and that give place to a real line bundle with the topology of an open M\"{o}bius strip.
Recall now that every vector bundle can be expressed as a subbundle of a trivial bundle of higher rank. In
the case of the M\"{o}bius bundle, this allows us to visualize the bundle in three dimensions. Consider the
trivial bundle over $S^1$ with total space $S^1\times \mathbb{R}^3$. The M\"{o}bius bundle can be described
as the subbundle of this trivial bundle whose fibre over  $\theta$ is the real line generated by the vector
$v(\theta)=(0,\sin(\theta/2),\cos(\theta/2))$. Note that, as we go around the circle, this vector rotates
through an angle of $\pi$, thus generating the twist of a M\"{o}bius strip. The description of this bundle by
means of a projector is as follows.  For each $\theta$ we can consider the projection from $\mathbb{R}^3$
onto the line generated by $v(\theta)$. Since the first component of $v(\theta)$ is always zero, we may
consider only the projection from the $y$-$z$ plane. A simple calculation shows that the matrix form of this
projector is
\begin{equation}
P(\theta)= \frac{1}{2}\left(\begin{array}{cc}
1-\cos\theta & \sin\theta\\
\sin\theta & 1+\cos\theta
\end{array}\right).
\end{equation}
Let $\mathcal R:=C(S^1,\mathbb{R})$ denote the algebra of continuos, real functions on the circle. Notice
that, although the components of $v$ do \emph{not} belong to  $\mathcal R$, the components of $P$ certainly
do. The projective $\mathcal R$-module associated to the bundle is thus given by $P(\mathcal R^2)$.
Explicitly, the elements of the module are vector valued maps $\sigma: S^1\mapsto \mathbb{R}^2$ of the form
$\sigma= a\sigma_1 + b \sigma_2$, where $a,b \in \mathcal R$ and $\sigma_1$ and $\sigma_2$ are the two
columns of $P$. This is precisely the space of sections of the  bundle the total space of which is given (as
a \emph{set}) by $\lbrace(\theta, \lambda v(\theta))\in S^1\times \mathbb{R}^3\,|\, \lambda\in
\mathbb{R}\rbrace$. Notice also that if we define  $\psi_a(\theta):=|\sin(\theta/2)|$ for $\theta\neq
0,2\pi$, with $\psi_a (0)\equiv \psi_a(2\pi):=0$ and, similarly, $\psi_b(\theta):=|\cos(\theta/2)|$ for
$\theta \neq \pi$, with $\psi_b(\pi):=0$, then $\mbox{supp}(\psi_{i})\subseteq U_{i}$ and $\psi_i\in \mathcal
R$ ($i=a,b$). One then checks that the components of $P$ are precisely given by the functions
$g_{ij}\psi_i\psi_j$ (with $i,j=a,b$ and $g_{ij}$ as in \eref{eq:A.1}), in accordance with remark
\ref{ap:A-remark}.
\end{example}
\section{The regular representation and projection operators}
\label{ap:B} Let $G$ be a finite group of order $|G|$. Denote with $\rho_i: G\rightarrow \mbox{Gl}(U^i)$ the
inequivalent irreducible representations of $G$, of dimension $n_i=\mbox{dim}\, U^i$, $i=1,\ldots,N$. If
$\mathcal F (G)$ denotes the vector space of complex functions on $G$, then the \emph{regular} representation
of $G$ is the one on $\mathcal F (G)$ which is induced by the group multiplication. It is a well-known fact
that the regular  representation contains all irreducible representations of $G$, each one with a
multiplicity equal to its dimension~\cite{Sternberg1994}:
\begin{equation}
\label{eq:B.1} \mathcal F (G)\cong\bigoplus_{i=1}^{N}n_i U^i.
\end{equation}
Consider now $\mathcal F (G)$ as a representation space for $G\times G$, where the representation $\hat r^G$
of $G\times G$ on $\mathcal F (G)$ is induced by the action
\begin{eqnarray}
(G\times G)\times G &\longrightarrow& G\nonumber\\
((g_1,g_2),h)&\longmapsto& g_1 h g_2^{-1}.
\end{eqnarray}
There is a close relation between the direct sum in (\ref{eq:B.1}) and the decomposition of $\hat r^G$ into
irreducible $G\times G$ representations. It is obtained in the following way. The representation $\rho_i$
induces a representation $\tilde\rho_i$ of $G$ on the dual space $U^{i*}$ (which is also irreducible and is
defined through $(\tilde\rho_i(g)\varphi)(u):=\varphi(\rho_i(g^{-1})u)$, with $\varphi\in U^{i*}$ and $u\in
U^i$) and it turns out that, as a $G\times G$ representation, $\rho_i\otimes \tilde\rho_i$ is irreducible.
Explicitly, we have:
\begin{eqnarray}
\rho_i\otimes \tilde\rho_i: G\times G &\longrightarrow& \mbox{Gl}(U^i\otimes U^{i*})\nonumber\\
\hspace{1.6cm}(a,b) &\longmapsto & \rho_i(a)\otimes \tilde\rho_i(b).
\end{eqnarray}
Irreducibility follows directly from the fact that $\tr (\rho_i(a)\otimes
\tilde\rho_i(b))=\tr\rho_i(a)\,\tr\tilde\rho_i(b)$. Consider now, for each $i$, the linear map $S_i:
U^i\otimes U^{i*}\rightarrow \mathcal F(G)$, defined on simple tensors by the formula $S_i(u\otimes
\varphi)(g):=\varphi(\rho_i(g^{-1})u)$, where $g\in G, u\in U^{i}$ and $\varphi\in U^{i*}$. It is easy to
see, using the orthogonality relations for the matrix elements of $\rho_i$, that the map $S_i$ is injective.
Moreover, it is $G\times G$-equivariant. This follows directly from the definitions given above since, for
$(a,b)\in G\times G$, $u\otimes\varphi\in U^i\otimes U^{i*}$ and $g\in G$, we have:
\begin{eqnarray}
S_i\left((a,b)\cdot u\otimes\varphi\right)(g) & \equiv &
              S_i\left(\rho_i(a)u\otimes\tilde\rho_i(b)\varphi\right)(g)\nonumber\\
   &=&  \big(\tilde\rho_i(b)\varphi\big)
\left(\rho_i(g^{-1})\rho_i(a)u\right)\nonumber\\
&=&  \varphi\left(\rho_i((a^{-1}gb)^{-1})u\right)\nonumber\\
  &=& S_i(u\otimes\varphi)((a,b)^{-1}\cdot g\nonumber)\\
 &=& \left(\hat{r}^G(a,b) S_i(u\otimes\varphi)\right)(g)\nonumber\\
 &\equiv& \left((a,b)\cdot S_i(u\otimes \varphi)\right)(g).
\end{eqnarray}
From this and the injectivity of $S_i$ it follows, from Schur's lemma, that the representations
 $(U^i\otimes U^{i*},\rho_i\otimes \tilde
\rho_i)$ and  $({\mathcal F}(G)|_{\mbox{Im}(S_i)},\hat{r}^G|_{\mbox{Im}(S_i)})$ must be equivalent. But then,
taking the sum of all $S_i$, we obtain a bijective linear map
\begin{equation}
S: \bigoplus_{i=1}^N(U^i\otimes U^{i*}) \longrightarrow \mathcal F(G)
\end{equation}
that furnishes an equivalence between $\big(\bigoplus_{i=1}^N (U^i\otimes U^{i*}),\bigoplus_{i=1}^N
(\rho_i\otimes \tilde \rho_i)\big)$ and  $({\mathcal F}(G),\hat r^G)$.

The isomorphism $S$ can be used to obtain an explicit formula for the projection from ${\mathcal F}(G)$ to
the copy of $U^i\otimes U^{i*}$ inside ${\mathcal F}(G)$, as explained below.  Chose a basis
$\{e^{\mt{(i)}}_r\}_r$ for $U^i$ and let $\{\tilde e^{\mt{(i)}}_r\}_r$ be the dual basis of $U^{i*}$, induced
by $\{e^{\mt{(i)}}_r\}_r$. Denote with $R^{\mt{(i)}}(g)$ the representing matrices of $\rho_i$, with respect
to $\{e^{\mt{(i)}}_r\}_r$:
\begin{equation}
\rho_i(g)(e^{\mt{(i)}}_r)=\sum_{r'=1}^{n_i} R^{\mt{(i)}}_{r',r}(g)e^{\mt{(i)}}_{r'}.
\end{equation}
Now, every function $f\in \mathcal F(G)$ can be written in the form $f=\sum_{g\in G} f(g)\delta_g$, where
$\delta_g$ is the characteristic function
\begin{equation}
\delta_g(h) := \left\{
\begin{array}{rl}
1 & \mbox{if } h=g,\\
0 & \mbox{if } h\neq g.
\end{array} \right.
\end{equation}
To obtain the projection operators, it is sufficient to compute $S^{-1}(\delta_g)$. We therefore seek the
(unique, $g$-dependent) coefficients $\lambda^i_{r,r'}$ such that
\begin{equation}
\label{eq:B.2} \delta_g= S\left( \sum_{i=1}^N\sum_{r,r'=1}^{n_i}\lambda^i_{r,r'} e^{\mt{(i)}}_r \otimes\tilde
e^{\mt{(i)}}_{r'}\right).
\end{equation}
Using the definition of $S$, we obtain:
\begin{eqnarray}
\hspace{-1cm} \delta_g(h) &=& S\left( \sum _{i=1}^N \sum_{r,r'=1}^{n_i}\lambda^i_{r,r'}
e^{\mt{(i)}}_r\otimes\tilde
e^{\mt{(i)}}_{r'} \right)(h)\nonumber\\
&=&   \sum _{i=1}^N \sum_{r,r'=1}^{n_i}\lambda^i_{r,r'} S_i(e^{\mt{(i)}}_r\otimes\tilde e^{\mt{(i)}}_{r'})
(h) =   \sum _{i=1}^N \sum_{r,r'=1}^{n_i}\lambda^i_{r,r'} \tilde
e^{\mt{(i)}}_{r'}(\rho_i(h^{-1})e^{\mt{(i)}}_r) \nonumber\\
&=&\sum _{i=1}^N \sum_{r,r',r''=1}^{n_i}\lambda^i_{r,r'}
 R^{\mt{(i)}}_{r'',r}(h^{-1})\tilde e^{\mt{(i)}}_{r'}(e^{\mt{(i)}}_{r''}) =
 \sum _{i=1}^N \sum_{r,r'=1}^{n_i}\lambda^i_{r,r'}R^{\mt{(i)}}_{r',r}(h^{-1}).
\end{eqnarray}
Multiplying the last equation by $R^{\mt{(j)}}_{k,l}(h)$ and summing over $h\in G$ we then obtain, using the
orthogonality relations:
\begin{equation}
\lambda^i_{r,r'}=\frac{n_i}{|G|}R^{\mt{(i)}}_{r,r'}(g).
\end{equation}
This implies, for $f\in \mathcal F(G)$,
\begin{equation}
\label{eq:B.4} f(g) = \sum_{h\in G} f(h) \delta_h(g)=\sum _{i=1}^N\sum_{r=1}^{n_i}\frac{n_i}{|G|} \sum_{h\in
G} R^{\mt{(i)}}_{r,r}(h^{-1})f(h^{-1}g).
\end{equation}
This shows that the operator $\frac{n_i}{|G|}\sum_{r=1}^{n_i}\sum_{h\in
G}R^{\mt{(i)}}_{r,r}(h^{-1})h\cdot(\,-\,)$ is the projection operator ${\mathcal F}(G)\rightarrow U^i\otimes
U^{i*}$. Moreover, since $U^i\otimes U^{i*}$ is isomorphic to the $n_i$-fold sum $U^i\oplus\cdots\oplus U^i$,
it is natural to consider the operators
\begin{equation}
\label{eq:B.3} P^{(i)}_{r,s}:=\frac{n_i}{|G|}\sum_{h\in G}R^{\mt{(i)}}_{r,s}(h^{-1})h\cdot(\,-\,).
\end{equation}
It turns out that these operators allow one to explicitly obtain the $r^{\mbox{\tiny th}}$ copy of $U^i$
inside $\mathcal F (G)$. \Eref{eq:B.3} is the starting point for the construction of a decomposition of the
algebra $C(M)$ into a direct sum of projective modules. The reader might have noticed that these operators
are also used in the field of Quantum Chemistry, e.g. for the construction of symmetry-adapted basis (cf.
\cite{Atkins1999}).

\section*{References}

\bibliographystyle{unsrt}

\end{document}